\documentclass[12pt]{article}
\usepackage{amsfonts}
\usepackage{amsmath}
\usepackage{amssymb}
\usepackage{amsthm}
\usepackage{setspace}
\usepackage{fullpage}
\usepackage{graphicx}
\usepackage{url}
\usepackage{verbatim}
\usepackage{enumerate}

\usepackage{tikz}
\usetikzlibrary{arrows}

  \newcommand{\tikzmath}[2][]
     {\vcenter{\hbox{\begin{tikzpicture}[#1]#2
                     \end{tikzpicture}}}
     }

\newcommand{\nc}[2]{\newcommand{#1}{#2}}
\nc{\C}{\mathbb{C}}
\nc{\R}{\mathbb{R}}
\nc{\Q}{\mathbb{Q}}
\nc{\Z}{\mathbb{Z}}
\nc{\N}{\mathbb{N}}
\nc{\cA}{\mathcal{A}}
\nc{\cB}{\mathcal{B}}

\nc{\g}{\mathfrak{g}}

\newcommand{\DeclareMyOperator}[1]{ \expandafter\DeclareMathOperator\csname #1\endcsname{#1} }

\def\CS{\mathit{C\hspace{-.4mm}S}}
\def\pt{pt}
\def\Rep{\mathrm{Rep}}

\def\Vect{\mathrm{Vect}}
\def\Bim{\mathrm{Bim}}

\def\op{\mathrm{op}}

\def\lnorm{
{\tikzmath{\useasboundingbox (-.12,-.1) rectangle (.13,.1);\node[scale=.8]{$\scriptscriptstyle \mathrm{l.n.\!\!}$};}}
}

\begin{document}
\title{What Chern-Simons theory assigns to a point}
\author{Andr\'e Henriques}
\date{}
\begin{singlespace}
\maketitle
\end{singlespace}

\newtheorem{theorem}{Theorem}%[section]
\newtheorem*{theorem*}{Theorem}
\newtheorem*{maintheorem}{Main Theorem}
\newtheorem{lemma}[theorem]{Lemma}
\newtheorem*{tech-lemma}{Technical lemma}
\newtheorem{proposition}[theorem]{Proposition}
\newtheorem{corollary}[theorem]{Corollary}
\newtheorem{definition}[theorem]{Definition}
\newtheorem{observation}[theorem]{Observation}
\newtheorem{exercise}[theorem]{Exercise}
\newtheorem{conjecture}[theorem]{Conjecture}
\newtheorem*{conjecture*}{Conjecture}
\newtheorem{question}[theorem]{Question}
\newtheorem*{question*}{Question}

\newenvironment{maindefinition}[1][Main Definition.]{\begin{trivlist}
\item[\hskip \labelsep {\bfseries #1}]}{\end{trivlist}}

\theoremstyle{remark}
\newtheorem{remark}[theorem]{Remark}
\newtheorem*{remark*}{Remark}
\newtheorem{example}[theorem]{Example}
\newtheorem*{example*}{Example}

\abstract{In this note, we answer the questions \emph{``What does Chern-Simons theory assign to a point?''}
and \emph{``What kind of mathematical object does Chern-Simons theory assign to a point?''}.\,

Our answer to the first question is \emph{representations of the based loop group}.
More precisely, we identify a certain class of projective unitary representations of the based loop group $\Omega G$ that we locally normal representations. %call positive energy representations.
We define the fusion product of such representations and we prove that, modulo certain conjectures,
the Drinfel'd centre of that representation category %locally normal %positive energy 
of $\Omega G$ is equivalent to the category of positive energy representations of the free loop group $LG$.
The above mentioned conjectures are known to hold when the gauge group is abelian or of type~$A_1$.

Our answer to the second question is \emph{bicommutant categories}.
The latter are higher categorical analogs of von Neumann algebras: they are tensor categories that are equivalent to their bicommutant inside %a certain fixed tensor category 
$\Bim(R)$, the category of bimodules over a hyperfinite $\mathit{III}_1$ factor.
We prove that, modulo certain conjectures, the category of locally normal %positive energy 
representations of the based loop group is a bicommutant category.
The relevant conjectures are known to hold when the gauge group is abelian or of type $A_n$.

Our work builds on the formalism of coordinate free conformal nets, developed jointly with A. Bartels and C. Douglas.
}

\tableofcontents

%\section{Introduction}

%\addtocounter{subsection}{-1}
\section{Chern-Simons theory}

The Chern-Simons theories are certain $3$-dimensional topological quantum field theories introduced by Witten \cite{MR990772}.
They are parametrised by a compact Lie group $G$ known as the \emph{gauge group}, and a cohomology class $k\in H^4(BG,\Z)$ known as the \emph{level} of the theory
\cite{MR1048699, MR1337109, MR1898192}. The Chern-Simons action %
%\footnote{The action is also often written as $S=\frac{1}{4\pi}\int_{M^3}\;\!\big\langle A\wedge dA\big\rangle_k+\tfrac23 \big\langle A\wedge A\wedge A\big\rangle_k$}
\begin{equation}\label{eq: Lagrangian of CS}
\qquad S=\frac{1}{4\pi}\int_{M^3}\;\!\big\langle A\wedge dA\big\rangle_k+\tfrac13 \big\langle A\wedge [A\wedge A]\big\rangle_k\qquad\text{\small $($mod  $2\pi)$}
\end{equation}
is a functional of $G$-bundles with connections over compact $3$-manifolds.
Here, $A$ is the connection form, $\langle\,\,\,\,\rangle_k:\g\otimes\g\to \R$ is a certain metric constructed from the level,
and the integral is taken over a global section of the principal bundle.\footnote{ When $G$ is not simply connected, principal bundles over $3$-manifolds can fail to have global sections and so
one cannot use the formula \eqref{eq: Lagrangian of CS} to define the action.
See
\cite{MR2174418, MR1048699, MR3330242}
for ways to overcome this difficulty.}
We should point out that not every level $k\in H^4(BG,\Z)$ yields a quantum field theory.
For example, it is important that $\langle\,\,\,\,\rangle_k$ be non-degenerate.
In this paper, we will only deal with those levels $k$ that satisfy the following positivity condition: %(see Section \ref{sec : H^4} for more details):

\begin{definition}\label{def: positive level}
Let $G$ be a compact Lie group.
A level $k\in H^4(BG,\Z)$ is positive if its image under the Chern-Weil homomorphism $H^4(BG)\to \mathrm{Sym}^2(\g^*)^G$ is a positive definite
symmetric bilinear form $\langle\,\,\,\,\rangle_k:\g\otimes\g\to \R$.
\end{definition}

We will write $\CS_{G,k}$ for the Chern-Simons theory associated to the gauge group $G$ and the level~$k$.
In the case of finite gauge groups, Chern-Simons theory is also known as Dijkgraaf--Witten theory.

It is well known that a $G$-bundle with connection is a critical point of the Chern-Simons action functional (a classical solution of the equations of motion) if and only if the connection is flat.
Such bundles are called local systems. % and
We write $\mathrm{Loc}_G(M)$ for the space of gauge equivalence classes of $G$-local systems on a manifold~$M$.\bigskip

From a mathematical point of view, the formula
\[
\CS_{G,k}(M)=\int_{\left\{\parbox{3.2cm}{\tiny Gauge equivalence classes of\newline $G$-bundles with connection $A$ on the $3$-manifold $M$}\right\}}e^{iS[A]}\,\,\mathcal DA
\]
used to `define' the value of the TQFT on a $3$-manifold $M$ (a.k.a. the value at $M$ of the partition function) is vacuous,
because the measure $\mathcal DA$ remains to be described.
However, the quantum Hilbert space $\CS_{G,k}(\Sigma)$ associated to a Riemann surface $\Sigma$ can and has been defined at a mathematical level of precision.
Following the prescription of \emph{canonical quantisation}, it is the geometric quantisation of $\mathrm{Loc}_G(\Sigma)$ with respect to the natural symplectic structure coming from the Chern-Simons Lagrangian 
\cite{MR1100212, MR1065677} 
(see e.g. \cite[\S6.1]{MR1701599} for a discussion of the symplectic structure).

Therefore, from a mathematical perspective, a TQFT that recovers the above quantum Hilbert spaces may be called `Chern-Simons theory'.
We will explain below that, at least when $G$ is simply connected (and presumably also when it is just connected), the Reshetikhin--Turaev TQFT associated to the
modular tensor category $\Rep^k(LG)$ of positive energy representations of the loop group at level $k$ has that property.

\section{Extended TQFTs}\label{sec: Extended TQFTs}
In the functorial approach to quantum field theory, a $d$-dimensional quantum field theory is a functor from a certain cobordism category,
whose objects are $(d-1)$-dimensional manifolds and whose morphisms are $d$-dimensional cobordisms,
to the category of vector spaces \cite{MR1001453, MR2079383}.
\emph{Extended quantum field theory} has been proposed by Lawrence \cite{MR1273575},
and later Freed \cite{MR1273572, MR1256993} and Baez-Dolan \cite{MR1355899},
as an enhancement of the functorial approach in which one assigns values to not only 
$d$- and $(d-1)$-dimensional manifolds, but also to $(d-2)$-dimensional, all the way down to $0$-dimensional manifolds.

In his influential paper \cite{MR1256993}
(see also \cite{MR1240583}),
Freed argued that Dijkgraaf--Witten theory fits into the framework of extended TQFT.
Using a `categorified path integral', he computed the value of that theory on the circle and showed that it is $\Rep(D^k(\C[G]))$,
the representation category of the $k$-twisted Drinfel'd double of the group algebra of $G$.
Freed did not extend the theory all the way down to points, even though the case $k=0$ of $\CS_{G,k}(\pt)$ is implicit in his paper.

The representation category of $D^k(\C[G])$ is equivalent to $\Vect^k_G[G]$,
the category of $k$-twisted equivariant vector bundles on $G$ with respect to the adjoint action of the group on itself \cite{MR2443249}.
The latter can in turn be identified with $Z(\Vect^k[G])$, the Drinfel'd centre
of the category of $G$-graded vector spaces with $k$-twisted convolution product~\cite{Wray-thesis}.

Summarising, for finite gauge group $G$, we have:
\begin{equation}\label{eq: Vect^k_G[G]}
\qquad\CS_{G,k}(S^1)=\Vect^k_G[G]=\Rep(D^k(\C[G]))=Z(\Vect^k[G]).
\end{equation}
The latter was taken as evidence in \cite{Wray-thesis} (see also  \cite{Waldorf-notes}) for the claim that
\begin{equation}\label{eq: CS(pt)=Vect^k[G]}
\qquad\CS_{G,k}(\pt)=\Vect^k[G] \qquad\qquad\text{\small ($G$ finite)}.
\end{equation}
Indeed, it is a general feature of TQFTs that for any manifold $M$ the value on $M\times S^1$ is the centre of the value on $M$
(where the meaning of `centre' depends on the context and in particular on the dimension of the TQFT).
As a special case, the value on $S^1$ should always be the centre of the value on a point.
See 
\cite[Lem. 3.75]{MR2713992} along with the discussion preceding that lemma for a proof in the case of $2$-dimensional TQFTs,
and see the proof of \cite[Prop. 4.9]{MR2648901} for a sketch in the case of $3$-dimensional TQFTs.

A more direct argument why $\Vect^k[G]$ deserves to be called $\CS_{G,k}(\pt)$ can be found in \cite{MR2648901}.
This goes via a certain 3-categorical limit construction which is a sort of discrete path integral.
The construction is described in \cite[\S8.1, ``case $m=2$'']{MR2648901} (compare with \cite[Example 3.14]{MR2648901} for more details on the ``case $m=1$'').\bigskip

When $G$ is a connected Lie group (always assumed compact), extended Chern-Simons theory is generally well understood down to dimension one.
In particular, it is widely agreed that the value on the circle should be the category $\Rep^k(LG)$ of positive energy representations of the loop group $LG$ at level $k$ (see e.g. \cite[\S4]{MR2476413}).
For example, it was shown by Freed and Teleman \cite{MR3334994} that $\Rep^k(LG)$ can be obtained directly from
from $\mathrm{Loc}_G(S^1)$ by a procedure akin to geometric quantisation.
There is a natural gerbe (the analog of a pre-quantum line bundle) on $\mathrm{Loc}_G(S^1)$ whose `sections' form a category.
Freed and Teleman identified a certain subcategory (the polarised sections) and proved that it is equivalent to $\Rep^k(LG)$.\footnote{Twisted loop groups are expected to show up when $G$ is disconnected \cite{MR2831111}.}

Let us present a different argument why
\begin{equation}\label{eq: CS_G,k(S^1)=Rep^k(LG)}
\CS_{G,k}(S^1)=\Rep^k(LG).
\end{equation}
Recall \cite{MR1100212} that the quantum Hilbert space $\CS_{G,k}(\Sigma)$ associated to a surface $\Sigma$ is the geometric quantisation of $\mathrm{Loc}_G(\Sigma)$.
In order the perform the geometric quantisation, one needs a polarisation of $\mathrm{Loc}_G(\Sigma)$, which requires a choice of complex structure on $\Sigma$.
So one gets a vector bundle over the moduli spaces of Riemann surfaces, which moreover comes with a natural projectively flat connection 
\cite{MR1100212, MR1065677}.
Starting instead from the right hand side\footnote{%
Here, we consider the incarnation of $\Rep^k(LG)$ as $\Rep^k(L\g)$, and we restrict to the case when $G$ is simply connected. See the next section for a discussion.} of \eqref{eq: CS_G,k(S^1)=Rep^k(LG)},
there is an associated complex modular functor \cite[Thms 5.7.11 and 6.7.12]{MR1797619},
which is part of the Reshetikhin--Turaev package \cite[Chap 4]{MR1797619}\cite{MR2654259}.
That modular functor is the same as the modular functor of WZW conformal blocks studied in \cite{MR1048605} (the two agree on genus zero curves by definition of the fusion product on $\Rep^k(L\g)$).
To summarise, both sides of \eqref{eq: CS_G,k(S^1)=Rep^k(LG)} are related to constructions of vector bundles with connection over the moduli spaces of Riemann surfaces:
one comes directly from the Chern-Simons Lagrangian,
and the other one comes from the modular tensor category $\Rep^k(LG)$.
Those vector bundles with connection are known to agree \cite{MR1669720} 
(see also
\cite{MR1289330, MR1257326, MR1289830, MR1456243}): that is our evidence for \eqref{eq: CS_G,k(S^1)=Rep^k(LG)}.\bigskip

Freed, Hopkins, Lurie, and Teleman \cite{MR2648901} have suggested to extend the proposal~\eqref{eq: CS(pt)=Vect^k[G]}
to the case when the gauge group is a torus $T$
by letting $\CS_{T,k}(\pt):=\mathrm{Sky}^k[T]$, the category of skyscraper sheaves on $T$ (that is, $T$-graded vector spaces with grading supported on finitely many points of $T$ and finite dimensional in each degree), with $k$-twisted associator.
As in \eqref{eq: Vect^k_G[G]}, the Drinfel'd centre\footnote{This holds for a certain variant of the Drinfel'd centre called `continuous Drinfel'd centre'.} of $\mathrm{Sky}^k[T]$ is equivalent to $\mathrm{Sky}^k_T[T]$.
But, unfortunately, this does not satisfy the desired property $Z(\mathrm{Sky}_\tau[T])=\Rep^k(LT)$.
It is however not too far off. One has
\[
Z(\mathrm{Sky}^k[T])=\mathrm{Sky}^k_T[T]\,\cong\, \Rep^k(LT)\otimes \mathrm{Sky}^k[\mathfrak t],
\]
and the extra factor $\mathrm{Sky}^k[\mathfrak t]$ can be interpreted as an anomaly.
Alternatively, the authors of \cite{MR2648901} claim that taking instead a certain \emph{relative} Drinfel'd centre of $\mathrm{Sky}_\tau[T]$ does recover $\Rep^k(LT)$.
We do not know of any functor between $\mathrm{Sky}^k[T]$ and our proposed answer for $\CS_{T,k}(\pt)$.

For connected non-abelian groups, the category $\mathrm{Sky}^k[G]$ still makes sense but its centre $\mathrm{Sky}^k_G[G]$ is very small (the conditions of $G$-equivariance and finite support are almost incompatible)
and does not seem very related to $\Rep^k(LG)$.

\section{Loop group representations} \label{I-sec 3}
As explained above, it is generally accepted that, for connected gauge groups, the value of Chern-Simons theory on the circle is $\Rep^k(LG)$, the category of positive energy representations of $LG$ at level $k$.
We recall the definition of that category.
First of all, a level $k\in H^4(BG,\Z)$ induces by transgression a central extension (see \cite[Thm 3.7]{MR2610397}\cite[Prop 5.1.3]{MR3055987})
\begin{equation}\label{eq: SES for LG}
1\,\to\, U(1)\,\to\, \widetilde{LG} \,\to\, LG \,\to\, 1
\end{equation}
of the free loop group $LG=\mathit{Map}(S^1,G)$. % (see Section \ref{sec: LG} for more details).  --------  PUT ME BACK ------  Add Waldorf reference?
The central extension depends on the level, but we suppress that dependence from the notation.
\begin{definition}[\cite{MR900587}]\label{def:RepLG}
Let $G$ be a connected Lie group (always assumed compact) and let $k\in H^4(BG,\Z)$ be a positive level.
A positive energy representation of $LG$ at level $k$ is a continuous\footnote{for the topology on $\widetilde{LG}$ induced by the $C^\infty$ topology on $LG$, and the strong operator topology on $U(H)$.} unitary representation 
$\widetilde{LG}\to U(H)$ that extends\footnote{The extension is never unique; the choice of extension is \emph{not} part of the data of a positive energy representation.} to an action of the semidirect product $S^1\ltimes \widetilde{LG}$
in such a way that the infinitesimal generator of the `energy circle' $S^1$ has positive spectrum and finite dimensional eigenspaces.
The category of positive energy representations of $LG$ at level $k$ is denoted $\Rep^k(LG)$.
\end{definition}
Here, the semidirect product is taken with respect to a certain action of $S^1$ on $\widetilde{LG}$ that lifts the `rotate the loops' action on $LG$.

The category $\Rep^k(LG)$ is known (or, depending on the definition, expected) to be a modular tensor category when equipped with the \emph{fusion product}.
People have historically preferred to work with a number of alternative categories (hopefully all equivalent).
These are (from now on, let us restrict to $G$ simple and simply connected% \footnote{%    ----------- PUT ME BACK
%In Section \ref{sec: LG}, we will explain how to extend the vertex algebra and conformal net approaches to the case of arbitrary connected Lie group $G$. We do not know whether something similar can be done with the approach via quantum groups.}
) the categories of representations of:
\begin{list}{}{\itemsep=-1ex \leftmargin=3.5ex \labelsep=1ex \topsep=-.5ex \parsep=2ex}
\item[\it a.] {\bf Quantum groups at root of unity \rm (see \cite{MR2286123} and references therein).}\\ We call the resulting braided tensor category $\Rep^\mathrm{ss}(U_q\g)$. Here, $q$ is the primitive ${m(k+h^\vee)}$\textsuperscript{th} root of unity, where $m\in\{1,2,3\}$ is the squared ratio of the lengths of the long roots to short roots. The superscript $^\mathrm{ss}$ stands for `semi-simplification' and refers to the operation of restricting to the subcategory of tilting modules, and then modding out by the negligible morphisms. 
The fusion product of representations comes from the Hopf algebra structure on $U_q\g$, and the universal $R$-matrix provides the braiding.
\item[\it b.] {\bf Affine Lie algebras/vertex operator algebras \rm (see \cite{MR3146015} and references therein).}\\
We call the resulting braided tensor category $\Rep^k(L\g)$.
It is the category of integrable highest weight modules of the affine Lie algebra, equivalently, of the corresponding vertex operator algebra (VOA).
The fusion product of representations is given indirectly\footnote{The fusion $W_1\boxtimes W_2$ can also be described directly \cite{MR1327541}, as the graded dual of a judiciously chosen subspace of the algebraic dual of $W_1\otimes W_2$. It is using that approach that the strongest results were obtained.}, by defining for each triple of representations $W_1$, $W_2$, $W_3$ the space of intertwinors 
from $W_1\boxtimes W_2$ to $W_3$.\footnote{See \cite[\S 9.3]{MR2082709} for the equivalence between the affine Lie algebra and VOA approaches.}
The associator is defined by means of the Kniznik--Zamolodchikov ODE over the moduli space of four-punctured spheres, and the braiding is defined similarly.
\item[\it c.] {\bf Conformal nets \rm \cite{MR1231644, MR1645078}\footnote{\cite{MR1231644} deals with all simply connected gauge groups, defines the braiding, but does not compute the fusion rules.
\cite{MR1645078} only deals with the case $G=SU(n)$, whose fusion rules it computes, but it does not discuss the braiding on the category of representations.} (see also
\cite[Sec 4.{\sc c}]{BDH-cn1} and references therein).%
}\\
We call the resulting braided tensor category $\Rep_{\mathrm f}(\cA_{G,k})$.%
\footnote{The subscript ${}_\mathrm{f}$ means that we only take representations which are \emph{finite} direct sums of irreducible ones.
We reserve the notation $\Rep(\cA_{G,k})$ for a category where infinite direct sums are allowed.}
Here, the conformal net $\cA_{G,k}$ assigns von Neumann algebras $\cA_{G,k}(I)$ to every interval $I\subset S^1$.
A representation is a Hilbert space with left actions of all those algebras.
The fusion of representations is based on
Connes' relative tensor product $\boxtimes$, also known as Connes fusion.\vspace{.1cm}
It is given by $H\boxtimes_{\cA_{G,k}(\tikz[scale=.3]{\useasboundingbox (-.5,-.2) rectangle (.5,.2);\draw (-.5,0) arc(-180:0:.5);\draw[->] (.15,-.5) -- +(.01,0);})}K$,
where the right action on $H$ uses the isomorphism
$\cA_{G,k}(\tikz[scale=.5]{\useasboundingbox (-.5,-.2) rectangle (.5,.2);\draw (-.5,0) arc(-180:0:.5);\draw[->] (.05,-.5) -- +(.01,0);})^\op\cong
\cA_{G,k}(\tikz[scale=.5]{\useasboundingbox (-.5,-.2) rectangle (.5,.2);\draw (-.5,0) arc(180:0:.5);\draw[->] (-.05,.5) -- +(-.01,0);})$
induced by reflection along the horizontal axis.
%The braiding is described in Section \ref{sec: Reps of nets}.       ---------- PUT ME BACK ----
\end{list}\vspace{3ex}

%\noindent
%Once can also define the relevant categories directly via {\it skein theory} (see [????????] and references therein).
%This is a combinatorial approach, where the categories are defined by generators and relations. 
%It has the advantage of being very explicit, but the construction has to be started anew for every family of Lie groups.

The approaches {\it a.} and {\it b.} have been well studied.
In particular, it is known that for every group $G$ and level $k\ge 0$ they describe modular tensor categories, and that they are additively equivalent to $\Rep^k(LG)$.
Moreover, the modular tensor categories obtained via {\it a.} and {\it b.} are known to be equivalent by combining the works of Finkelberg \cite{MR1384612, MR3053762}
and of Kazhdan--Lusztig (\cite{MR1186962, MR1239506, MR1239507, MR1104840} for the simply laced case; \cite{MR1276910, MR2105507} for the non simply laced case;
the exceptional cases $E_6$, $E_7$, $E_8$ level $1$ and $E_8$ level 2, where the results of Kazhdan and Lusztig do not apply, require an ad hoc analysis~\cite{MO:178113}).
%
%In the special case of the group $G=SU(n)$, there is also the work of Andersen--Ueno
%\cite{MR2339577, MR2306213, MR2928086, Andersen-Ueno(Reshetikhin-Turaev-from-CFT)}, who
%directly identify the WZW modular functor of affine $\mathfrak{sl}(n)$ with the one 
%associated to 
%the skein theory model of Blanchet \cite{MR1710999}.             % ------ REMOVE?

The approach {\it c.\!} is less developed.
So far, only the following results appear to be know: for the group $G=SU(n)$, the braided tensor category is modular \cite{MR1838752, MR1776984},
and its fusion rules agree with those of the corresponding modular tensor categories constructed via {\it a.\!} and {\it b.\!} \cite[\S 34]{MR1645078}\footnote{The work of Wassermann does not exclude the possibility of `exotic' representations of $\cA_{SU(n),k}$, that do not come from representation of the affine Lie algebra.
Those can indeed be excluded by combining \cite[Thm 3.5.1 + eq. on line 2 of p.18]{MR1776984} with \cite[Thm 33 + Cor 39]{MR1838752}.} %, or by our Proposition (...) when $n\ge 3$.}  ------ PUT ME BACK ----
(the latter are well known \cite{MR1328736, MR2286123}, see also \cite{MR1360497, MR1257326, MR1456243}\,\cite[\S 7.3]{MR1797619}).
Even for $G=SU(n)$, the categories constructed via {\it a.} (or {\it b.}) and {\it c.} are not known to be equivalent as  braided tensor categories, unless $n=2$.
For other Lie groups, the braided tensor category {\it c.\!} is not known to be fusion (e.g., the tensor product multiplicities are not known to be finite)
and also not known to be additively equivalent to the one constructed via {\it a.} or {\it b.} % --- see (...) and (...) for a partial result.   ---------PUT ME BACK --------
Despite all the above, the following conjecture is widely believed to be true:

\begin{conjecture} \label{conj: 1}
For every simple simply connected Lie group $G$, and every level $k\ge 0$, 
the categories $\Rep^k(LG)$ defined via {\it a.} (or {\it b.}), and {\it c.} are equivalent as balanced tensor categories\footnote{A balanced tensor category is a tensor category with a braiding and a twist \cite{MR1107651}.}.
\end{conjecture}

For $G=SU(2)$, the above conjecture can be proved as follows.\footnote{This argument, as well as the one below for tori, was communicated by Marcel Bischoff.} 
As mentioned earlier, the fusion rules are known to agree by the work of Wassermann.
By \cite[Prop. 8.2.6]{MR1239440}\footnote{That proposition only applies to the levels $k\ge 2$. See p.387 of that same reference for a discussion of the (easier) case $k=1$.},
balanced tensor categories with $SU(2)$ level $k$ fusion rules are determined by the entires of their $T$-matrix\footnote{also known as conformal spins, or balancing phases.}.
The latter are the exponentials of the conformal weights, both
when the modular tensor category comes from a VOA \cite[Thm 4.1]{MR2468370} and when it comes from a conformal net \cite{MR1410566}, and therefore agree.

In the case when $G$ is connected but not simply connected, the vertex operator algebra and conformal net approaches still make sense \cite{WZW-classification}:
the VOA/conformal net associated to a connected Lie group is a simple current extension of the tensor product of one associated to an even lattice and one associated to a simply connected Lie group.
%(see Definitions \ref{def: loop group conf net} and \ref{def: loop group VOA}). ---- PUT ME BACK -----
The above conjecture can thus be generalised:

\begin{conjecture} \label{conj: 2}
For every connected Lie group $G$ and every positive level $k\in H^4(BG,\Z)$, 
the categories $\Rep^k(LG)$ defined via vertex operator algebras and via conformal nets are equivalent as balanced tensor categories.
\end{conjecture}

When the gauge group is a torus, Conjecture~\ref{conj: 2} follows from known computations on the VOA side 
\cite[Chapt 12]{MR1233387}\,\cite{MR1245855} 
and the conformal net side
\cite{MR2261756, Staszkiewicz-thesis},
because a modular tensor category all of whose objects are invertible
is entirely determined by its fusion rules and by the entries of its $T$-matrix \cite[Prop 2.14]{MR2076134}.
Conjecture~\ref{conj: 2} actually follows from Conjecture~\ref{conj: 1} and the case of tori:
the representation category of an extension can be described entirely in terms of 
the representation category of the original vertex operator algebra \cite[Thm 3.4]{MR3339173} or conformal net \cite[Prop 6.3]{arXiv:1410.8848}.
So the two conjectures are equivalent.
\bigskip

%An alternative to the approaches {\it a.}, {\it b.} and {\it c.} explained above
%is to follow \cite{MR900587} and define $\Rep^k(LG)$ to be the category of positive energy (projective) unitary representations of the smooth loop group $LG:=\mathit{Map}(S^1,G)$.
%Unfortunately, to our knowledge, nobody has managed to define the monoidal structure on $\Rep^k(LG)$ using that approach.

\noindent
{\it Digression.} We now explain an attempt due to Graeme Segal at defining the fusion product directly on $\Rep^k(LG)$.
We follow and expand the discussion in \cite[\S5]{MR2079383}.
Given a pair of pants $\Sigma$ with complex structure in the bulk and analytically parametrised boundary $\partial \Sigma=S_1\cup S_2\cup S_3$,
one ought to be able to define the fusion product as follows.\footnote{The fusion product will depend on $\Sigma$.
Given $\Sigma_1$ and $\Sigma_2$, the corresponding bifunctors $\boxtimes_1$ and $\boxtimes_2$ will be equivalent, but non-canonically so.
In other words, $H_\lambda\boxtimes_1H_\mu$ and $H_\lambda\boxtimes_2H_\mu$ will be isomorphic but with no given isomorphism $H_\lambda\boxtimes_1H_\mu\to H_\lambda\boxtimes_2H_\mu$, unless additional data is provided.}
Let $H_\lambda$ and $H_\mu$ be two unitary positive energy representations of $LG$,
and let $\check H_\lambda$ and $\check H_\mu$ be their dense subspaces of analytic vectors for the rotation action of $S^1$.
On those subspaces, the projective action of the loop group extends to its complexification $LG_\C:=\mathit{Map}(S^1,G_\C)$.
Let $\mathrm{Bun}_{G}(\Sigma;S_3)$ be the moduli space of holomorphic $G_\C$-bundles over $\Sigma$, trivialised over $S_3$.
Similarly, let $\mathrm{Bun}_{G}(\Sigma;\partial\Sigma)$ be the moduli space of holomorphic bundles trivialised over the whole boundary,
and note that $\mathrm{Bun}_{G}(\Sigma;S_3)$ is the quotient of $\mathrm{Bun}_{G}(\Sigma;\partial \Sigma)$ by the gauge action of $LG_\C\times LG_\C$ at $S_1$ and $S_2$.
%Equivalently, this is the quotient of $\mathrm{Bun}_{G}(\Sigma;\partial \Sigma)$ by the gauge action of $LG_\C\times LG_\C$ at $S_1$ and $S_2$.
Let $\widetilde{\mathrm{Bun}_{G}(\Sigma;\partial \Sigma)}$ be the $k$-th power of the determinant bundle over $\mathrm{Bun}_{G}(\Sigma;\partial \Sigma)$, minus its zero section
(see \cite{MR1690736} for a definition).
The action of $LG_\C\times LG_\C$ on $\mathrm{Bun}_{G}(\Sigma;\partial \Sigma)$ then lifts to an action of $\widetilde{LG_\C\times LG_\C}:=\widetilde{LG}_\C\times_{\C^\times} \widetilde{LG}_\C$ on $\widetilde{\mathrm{Bun}_{G}(\Sigma;\partial \Sigma)}$, whose quotient is again $\mathrm{Bun}_{G}(\Sigma;S_3)$.
So we may consider the vector bundle
\[
\mathrm{H}_{\lambda\mu}:=\widetilde{\mathrm{Bun}_{G}(\Sigma;\partial \Sigma)}\times_{\widetilde{LG_\C\times LG_\C}} \check H_\lambda\otimes \check H_\mu
\]
over $\mathrm{Bun}_{G}(\Sigma;S_3)=\mathrm{Bun}_{G}(\Sigma;\partial \Sigma)/LG_\C\times LG_\C$.
Segal's hope is that one should be able to define the fusion product $H_\lambda\boxtimes H_\mu$ as an appropriate completion of the space
\begin{equation}\label{eq: too big}
\Gamma_{\lambda\mu}:=\Gamma^{hol}\big(\mathrm{Bun}_{G}(\Sigma;S_3),\mathrm{H}_{\lambda\mu}\big)
\end{equation}
of holomorphic sections of $\mathrm{H}_{\lambda\mu}$ (or rather, a completion of a dense subspace thereof).\footnote{An alternative proposal is to use $\Gamma^{hol}(\mathrm{Bun}_{G}(\Sigma;S_3),\mathrm{H}_{\lambda\mu}')'$, where the prime means dual topological vector space (bundle). Whereas the first proposal \eqref{eq: too big} is in some sense `too big', the second one is possibly `too small'.}
The gauge action at $S_3$ of $LG_\C$ on $\mathrm{Bun}_{G}(\Sigma;\partial \Sigma)$
lifts to an action of the central extension $\widetilde{LG}_\C$ on $\widetilde{\mathrm{Bun}_{G}(\Sigma;\partial \Sigma)}$.
The latter yields an action of $\widetilde{LG}_\C$ on $\mathrm{H}_{\lambda\mu}$, and therefore on its space of holomorphic sections.

There are many difficulties in implementing the above ideas. Fist of all, it is not clear how to endow $\Gamma_{\lambda\mu}$ with an inner product (probably only defined on a dense subspace)
and therefore not clear how to perform the completion.
Most importantly, it is not clear that $\Gamma_{\lambda\mu}$ is a positive energy representation.
In other words, it is not clear that the action of $\widetilde{LG}$ extends to the semidirect product $S^1\ltimes \widetilde{LG}$ (most likely, the action extends to $S^1\ltimes \widetilde{LG}$ only on a dense subspace of $\Gamma_{\lambda\mu}$).\footnote{We
thank Graeme Segal for explaining to us the state of his program, and the difficulties encountered.}
When $G$ is a torus, the above difficulties were overcome by Posthuma \cite{MR2925299}.
The general case was thought to have been treated in \cite{Posthuma(PhD-thesis)}, but a mistake was later found in \cite[Prop.~5.5]{Posthuma(PhD-thesis)}\footnote{We thank Hessel Posthuma for clarifications on this point.}.

\section{The value on a point: $\Rep^k(\Omega G)$}  \label{I-sec 4}
Let $G$ be a connected Lie group.
It is a general feature of $3$-dimensional TQFTs that the value on $S^1$ is the Drinfel'd centre of the value on a point.
%Recall also that it is generally understood that $\CS_{G,k}(S^1)=\Rep^k(LG)$.

Thus, following  \cite[\S4]{MR2476413},
we take the point of view that a tensor category $T=T_{G,k}$ deserves to be called the value of Chern-Simons theory on a point if its Drinfel'd centre $Z(T)$ is equivalent to $\Rep^k(LG)$. 
The question \emph{`What does Chern-Simons theory assign to a point?'} therefore reduces to:

\begin{question*}
Find a tensor category $T_{G,k}$ whose Drinfel'd centre $Z(T_{G,k})$ is equivalent to the category $\Rep^k(LG)$
of positive energy representations of the loop group at level~$k$.
\end{question*}

We will argue that the category of representations of the based loop group at level $k$ offers an answer to the above question, and hence deserves to be called the value of Chern-Simons theory on a point:
\begin{equation*} %\label{kjksfnlbnfgbnfl}
Z\big(\Rep^k(\Omega G)\big)=\Rep^k(LG)
\end{equation*}
As explained in the previous section, there are more than one possible meanings for $\Rep^k(LG)$. So we need to be a little bit more precise.

Let $H_0\in\Rep^k(LG)$ be the vacuum representation of $LG$ at level $k$ (the unit for the fusion product).
Given an interval $I\subset S^1$, we write $L_IG\subset LG$ for the subgroup of loops with support in that interval,
we write $\widetilde{L_I G}$ for the corresponding central extension (induced by \eqref{eq: SES for LG}),
and $\cA_{G,k}(I)$ for the von Neumann algebra generated by $\widetilde{L_I G}$ inside the bounded operators on $H_0$.
The assignment $\cA_{G,k}:I\mapsto \cA_{G,k}(I)$ is a conformal net.
From now on, we take $\Rep_{\mathrm f}(\cA_{G,k})$ as our working definition of $\Rep^k(LG)$.
We will actually be working with the slightly larger category $\Rep(\cA_{G,k})=\Rep_{\mathrm f}(\cA_{G,k})\otimes_{\mathsf{Vec}_{\mathrm {f\!\!\;.\!\!\;d\!\!\;.}}}\mathsf{Hilb}$ (\cite[\S3.2]{colimits}).
The latter can also be described as the category of unitary representations of $\widetilde{LG}$
with the property that, for every interval $I \subset S^1$, the action of $\widetilde{L_I G}$ extends to an action of $\cA_{G,k}(I)$ \cite[Thm.\,26]{colimits}.
We call such representations \emph{locally normal}, and we denote them $\Rep^k_\lnorm(LG)$.

%In this paper, assuming Conjectures \ref{conj: 1} and \ref{conj: 2}, w
Assuming Conjectures \ref{conj: 1} and \ref{conj: 2},
we will prove that, 
provided one replaces $\Rep_{\mathrm f}(\cA_{G,k})$ by the slightly larger category $\Rep^k_\lnorm(LG)%=\Rep^k(LG)\otimes_{\mathsf{Vec}_{\mathrm {f\!\!\;.\!\!\;d\!\!\;.}}}\mathsf{Hilb}
=\Rep(\cA_{G,k})$, the category %$T_{G,k}:=\Rep^k_\lnorm(\Omega G)$ 
of locally normal representations of the based loop group at level $k$ satisfies
\begin{equation} \label{kjksfnlbnfgbnfl}
Z\big(\Rep^k_\lnorm(\Omega G)\big)=\Rep^k_\lnorm(LG).
\end{equation}
We take this as evidence for our claim that $\Rep^k_\lnorm(LG)$
%and therefore deserves to be called 
is the value of Chern-Simons theory on a point.
(Of course, if one adopts the Platonic point of view that $\CS_{G,k}(\pt)$ is something that \emph{exists} and whose value we are trying to compute,
then our reasoning according to which, since
%\footnote{Equation \eqref{kjksfnlbnfgbnfl} is actually not quite right. The correct version is $Z(\Rep^k(\Omega G))=\Rep^k(LG)\otimes_{\mathsf{Vec}_{\mathrm {f\!\!\;.\!\!\;d\!\!\;.}}}\mathsf{Hilb}$, i.e., one needs to remove the condition that the energy operator has finite dimensional eigenspaces in Definition~\ref{def:RepLG}.}
equation \eqref{kjksfnlbnfgbnfl} holds, it must be the case that
\begin{equation}\label{eq: CS(pt)=Rep(OmG)}
\,\qquad\qquad\qquad\qquad 
\CS_{G,k}(\pt)=\Rep^k_\lnorm(\Omega G)
\qquad\quad\text{\footnotesize ($G$ connected)}
\end{equation}
does not constitute a proof, as there might exist other tensor categories with the same Drinfel'd centre.)

%Let us now define positive energy representations of the based loop group at level $k$.
Let $\Omega G\subset LG$ be the subgroup consisting of loops $\gamma:S^1\to G$ such that $\gamma(1)=e$ and such that all the higher derivatives of $\gamma$ vanish at that point, and
let $\widetilde{\Omega G}$ be the corresponding central extension, inherited from \eqref{eq: SES for LG}.
%be the restriction of the central extension \eqref{eq: SES for LG} to the subgroup $\Omega G\subset LG$ .
 % and
%let $H_0\in\Rep^k(LG)$ be the vacuum representation of $LG$ at level $k$ (the unit for the fusion product).
%Given an interval $I\subset S^1$, we write $L_IG\subset LG$ for the subgroup of loops with support in $I$,
%and $\widetilde{L_I G}$ %\subset \widetilde{LG}$ 
%for the corresponding central extension.
%Finally, we let $\cA_{G,k}(I)$ be the von Neumann algebra generated by $\widetilde{L_I G}$ inside the algebra of bounded operators on $H_0$.

\begin{list}{}{\itemsep=-1ex \leftmargin=3.5ex \labelsep=1ex \topsep=2ex \parsep=2ex \rightmargin=3.5ex}
\item{\bf Main definition.}
A unitary representation of $\widetilde{\Omega G}$ on a Hilbert space is a \emph{locally normal representation of $\Omega G$ at level $k$} if for every interval $I \subset S^1$ such that the base point $1\in S^1$ is not in the interior of $I$, the 
action of $\widetilde{L_I G}$ extends to an action of the von Neumann algebra $\cA(I):=\cA_{G,k}(I)$.
We write $\Rep^k_\lnorm(\Omega G)$ for the category of locally normal representations of $\Omega G$ at level $k$.
The monoidal structure on $\Rep^k_\lnorm(\Omega G)$ is given by
\[
(H,K)\mapsto H\boxtimes_{\cA(\tikz[scale=.3]{\useasboundingbox (-.5,-.2) rectangle (.5,.2);\draw (-.5,0) arc(-180:0:.5);\draw[->] (.15,-.5) -- +(.01,0);})}K
\]
as in the definition of fusion of representations of conformal nets.
The actions of
$\cA(\tikz[scale=.5]{\useasboundingbox (-.5,-.2) rectangle (.5,.2);\draw (-.5,0) arc(-180:0:.5);\draw[->] (.05,-.5) -- +(.01,0);})$ on $H$ and of
$\cA(\tikz[scale=.5]{\useasboundingbox (-.5,-.2) rectangle (.5,.2);\draw (-.5,0) arc(180:0:.5);\draw[->] (-.05,.5) -- +(-.01,0);})$ on $K$ induce actions of those same algebras on the fused Hilbert space, which in turn uniquely extend to an action of $\widetilde{\Omega G}$ \cite[Lem.\,4.4]{Bicommutant-categories-from-conformal-nets}\cite[Thm.\,31]{colimits}.\medskip
\end{list}

Note that since the algebras $\cA_{G,k}(I)$ are type $\mathit{III}$ factors \cite[Thm 2.13]{MR1231644}\,\cite[Prop.\,1.2]{MR1410566} \cite[Prop 6.2.9]{Longo(Lectures-on-Nets-II)},
the following is an equivalent\footnote{Provided one restricts to separable Hilbert spaces.} description of the category $\Rep^k(\Omega G)$:
a representation of the based loop group is locally normal if it is either zero or, for every interval $I\subset S^1$ such that the base point is not in the interior of $I$, its restriction to $\widetilde{L_I G}$ is equivalent to $H_0$, the vacuum representation at level $k$.

We conjecture that our notion of locally normal representation of $\Omega G$ at level $k$ admits the following alternative description:
a Hilbert space with a strongly continuous action of $\widetilde{\Omega G}$ that extends to $\R\ltimes\widetilde{\Omega G}$ in such a way that the spectrum of the infinitesimal generator of $\R$ is positive.
Here, the semi-direct product is taken with respect to loop reparametrisations by M\"obius transformations that fix the base point $1\in S^1$.
%In Proposition (...), we will construct a functor going in one direction.   ------ PUT ME BACK ---
Such a description would be attractive because directly parallel to the classical definition \cite{MR900587} of positive energy representation of~$LG$ at level $k$.

\begin{remark}
One possible objection to out proposal \eqref{eq: CS(pt)=Rep(OmG)} is that the category $\Rep^k_\lnorm(\Omega G)$ is too big
(it is neither rigid nor even abelian; in particular, it falls completely outside of the framework of \cite{arXiv:1312.7188}).
This is however completely unavoidable. 
By the results of \cite[\S 5.5]{MR3039775} (see also \cite{MR1966525}),
if there is a fusion category $T$ such that $Z(T)\cong \Rep^k(LG)$, then the latter must have central charge divisible by~$8$, a fact which only holds in very few cases.
\end{remark}

\section{Conformal nets and bicommutant categories}  \label{I-sec 5}

The motivating examples come from loop groups,
but our results apply to any conformal net with appropriate finiteness conditions.

For every conformal net $\cA$, we consider the following tensor category $T_\cA$.
The objects of $T_\cA$ are Hilbert spaces equipped with compatible actions of the algebras $\cA(I)$ for every 
interval $I \subset S^1$ such that the base point of $1\in S^1$ is not in the interior of $I$.
Such representations are known in the conformal net literature as \emph{solitons}
\cite{MR1652746, MR1892455, MR1332979, MR2100058}.
When $\cA$ is $\cA_{G,k}$, the conformal net associated to a loop group, this recovers the category $\Rep^k_\lnorm(\Omega G)$ of locally normal representations of the based loop group \cite[Thm.\,31]{colimits}.

Recall \cite{MR1838752} that there is a certain invariant $\mu(\cA)\in \R_+\cup\{\infty\}$ of a conformal net, called the $\mu$-index.
By definition, it
is the Jones--Kosaki index \cite{MR696688, MR829381} of the subfactor \smallskip
\[
\cA(I_1)\vee \cA(I_3)\subset (\cA(I_2)\vee \cA(I_4))',\smallskip
\]
where $I_1, I_2, I_3, I_4$ are intervals that cover the circle as follows:
\tikz[scale=.5]{
\useasboundingbox (-1.2,-.1) rectangle (1.2,.8);
\draw (0:1)+(.1,.1) arc (0:90:1)(0:1)+(.1,-.1) arc (0:-90:1)(0:-1)+(-.1,.1) arc (180:90:1)(0:-1)+(-.1,-.1) arc (-180:-90:1);
\node at (1.05,1) {$\scriptstyle I_1$};\node at (-1.05,1) {$\scriptstyle I_2$};\node at (-1,-1) {$\scriptstyle I_3$};\node at (1.05,-1) {$\scriptstyle I_4$};}
and the prime denotes the\hspace{-.1mm} commutant,\hspace{-.1mm} taken\hspace{-.1mm} on\hspace{-.1mm} the\hspace{-.1mm} vacuum\hspace{-.1mm} sector\hspace{-.1mm} of\hspace{-.1mm} the\hspace{-.1mm} con-\hspace{1.6cm}formal net.
It has the property \cite{MR1838752, MR2100058} that\footnote{The last two conditions in \eqref{eq: 3 equivalent defs of finite mu} are phrased in a somewhat sloppy way.
The correct way to formulate them is to say that $\Rep(\cA)\cong \mathcal C\otimes_{\mathsf{Vec}_{\mathrm {f\!\!\;.\!\!\;d\!\!\;.}}}\mathsf{Hilb}$ for some fusion (modular) tensor category~$\mathcal C$.}
\begin{equation}\label{eq: 3 equivalent defs of finite mu}
\mu(\cA)<\infty
\quad\Leftrightarrow\quad
\Rep(\cA)\text{ is fusion}
\quad\Leftrightarrow\quad
\Rep(\cA)\text{ is modular}
\end{equation}
(see \cite[Sec. 3]{BDH-cn1}\cite[Sec. 3C]{BDH-cn2} for an alternative proof of this fact).

Our main result says that, for conformal nets with finite $\mu$-index, the
Drinfel'd centre of $T_\cA$ is equivalent to the category of representations of the conformal net:

\begin{maintheorem}\label{thm: If Rep(cA_G,k) is modular}
If $\cA$ has finite $\mu$-index, then there is an equivalence of balanced tensor categories $Z(T_\cA)\cong\Rep(\cA)$.
\end{maintheorem}

\noindent
The proof of this theorem is the content of our companion paper \cite{Bicommutant-categories-from-conformal-nets}
(see \cite[Rem.\,10]{Bicommutant-categories-from-conformal-nets} for a discussion of the balanced structure on $Z(T_\cA)$ and on $\Rep(\cA)$).\footnote{
The braiding on $\Rep(\cA)$ defined in \cite[Sec. 3B]{BDH-cn2} has not been compared to the one used in \cite{MR1838752}.
We can therefore not exclude the possibility that, when $\mu(\cA)<\infty$, the category $\Rep(\cA)$ has two \emph{distinct} modular structures.
In \cite{Bicommutant-categories-from-conformal-nets}, we prove our main theorem for the braided structure \cite{MR1838752}.}

It is widely expected that the conformal nets associated to loop groups have finite $\mu$-index (in which case the above theorem could be applied to them), but this remains an open problem.
At the moment, the best result in that direction is the one of Xu \cite{MR1776984}, based on the work of Wassermann \cite{MR1645078}, according to which the conformal nets associated to $SU(n)$ have finite $\mu$-index.
\bigskip

Going back to the special case of Chern-Simons theory, and using the fact that $\Rep(\cA_{G,k})=\Rep^k_\lnorm(LG)$, we have the following corollaries of our main theorem:
\begin{corollary}\label{thm:3}
For $G=SU(n)$, the Drinfel'd centre $Z(\Rep^k_\lnorm(\Omega G))$ of the category of locally normal representations of the based loop group at level $k$
is equivalent as a balanced tensor category to $\Rep^k_\lnorm(LG)$, the category of locally normal representations of the free loop group at level $k$.
\end{corollary}

\begin{proof}
$\mu(\cA_{SU(n),k})<\infty$ by \cite{MR1776984}.
\end{proof}

\begin{corollary}\label{thm:4}
For every connected Lie group $G$ for which Conjecture \ref{conj: 2} holds, there is an equivalence of balanced tensor categories 
between $Z(\Rep^k_\lnorm(\Omega G))$ and the category of positive energy representations of the free loop group at level $k$ (Definition~\ref{def:RepLG}),
provided one removes the condition that the energy operator has finite dimensional eigenspaces.
\end{corollary}

\begin{proof}
Conjecture \ref{conj: 2} implies $\mu(\cA_{G,k})<\infty$ because of the equivalence \eqref{eq: 3 equivalent defs of finite mu} and of Huang's theorem
\cite{MR2468370, MR3339173}, according to which the representation categories of the relevant vertex operator algebras are modular.
\end{proof}

By Lurie's classification of extended topological field theories 
\cite[Thm 1.4.9]{MR2555928},
the map $Z\mapsto Z(pt)$ which sends a TQFT $Z$ to its value on the point provides a bijection between extended $n$-dimensional TQFTs and fully dualisable objects in the given target $n$-category.
From the remark at the end of the previous section, it might look like $T_\cA$ is not fully dualisable,
which would imply that there is no TQFT whose value on a point is $T_\cA$.
We believe that it is possible to restore the full dualisability of $T_\cA$ by viewing it as an object not of the $3$-category of tensor categories,
but of a yet to be constructed $3$-category of \emph{bicommutant categories}. %\footnote{Constructing the $3$-category of bicommutant categories will probably require quite some effort. We do not claim to have done this.}
%At this point, we only know how to describe its objects, 1-morphisms, 2-morphisms, and 3-morphisms. We do not have a general definition of composition of 1-morphisms.}

Let $R$ be a hyperfinite $\mathit{III}_1$ factor, and let $\Bim(R)$ denote its bimodule category, equipped with Connes' relative tensor product.
The latter comes with antilinear involutions 
at the level of objects (the contragredient of a bimodule)
and at the level of morphisms (the adjoint of a linear map).

\begin{definition}[{\cite[\S3]{Bicommutant-categories-from-fusion-categories}}]
A bicommutant category is a tensor category $T$ equipped with two involutions as above, and a tensor functor $T\to \Bim(R)$ (often a fully faithful embedding), compatible with the two involutions, so that the natural map $T\to T''$ of the category to its bicommutant is an equivalence.
\end{definition}

Here, we write $T'$ for the commutant of the tensor category $T$.
It is the category whose objects are pairs $(Y,e)$ with $Y\in\Bim(R)$ and $e=\{e_X:X\boxtimes Y\to Y\boxtimes X\}_{X\in T}$ a half-braiding with all the elements of $T$
(which we abusively identify with their image in $\Bim(R)$). %, and whose morphisms are morphisms $X_1\to X_2$ in $\Bim(R)$ that are compatible with the half-braiding.
The bicommutant $T'':=(T')'$ is equipped with a natural `inclusion' functor $T\to T''$. % $T\to T'':X\mapsto (X,\{e_X^{-1}:Y\boxtimes X\to X\boxtimes Y\}_{(Y,e)\in T'})$.
%The following is the other main result of this paper:

\begin{theorem}
If $\cA$ is a conformal net with finite $\mu$-index, then $T_{\cA}$ is a bicommutant category.
\end{theorem}

\noindent
The proof of this theorem can be found in our companion paper \cite{Bicommutant-categories-from-conformal-nets}.

\begin{remark*}
In our earlier paper \cite{BDH-cn0}, we had suggested using the $3$-category of conformal nets, constructed in \cite{BDH-cn1, BDH-cn3, BDH-cn4},
as a target category for extended 3-dimensional TQFTs, and
to have $\cA_{G,k}$ be the value of Chern-Simons theory on a point.
We conjecture that the construction $\cA\,\mapsto\, T_{\cA}$ extends to a fully faithful but maybe not essentially surjective 3-functor
\[
T\,\,:\,\,\,
\{\text{conformal nets}\}\to\{\text{bicommutant categories}\}.
\]
Such a functor would make our current proposal \eqref{eq: CS(pt)=Rep(OmG)} for the value of Chern-Simons theory on a point
`backwards compatible' with respect to our earlier proposal \cite{BDH-cn0}.
\end{remark*}

\section{$\CS(\pt)$ for disconnected groups}

In Sections \ref{I-sec 3}--\ref{I-sec 5}, the gauge group was always connected.
Let now $G$ be an arbitrary compact Lie group, and let $k\in H^4(BG,\Z)$ be a positive level.
We propose a general answer to the question of what Chern-Simons theory assigns to a point,
that simultaneously generalises the previous answers \eqref{eq: CS(pt)=Vect^k[G]} and \eqref{eq: CS(pt)=Rep(OmG)} in the cases of finite gauge group and connected gauge group respectively.

Let $\mathrm{Bun}_G(S^1;*)$ be the moduli space of $G$-bundles over $S^1$ trivialised to infinite order at the base point $*=1\in S^1$.
This stack has finitely many points (the isomorphism classes of $G$-bundles over $S^1$) classified by their monodromy in $\pi_0(G)$, and each point has an infinite dimensional isotropy group (the automorphism group of the $G$-bundle) which is isomorphic to $\Omega G$.
For ease of notation, we fix for every $[g]\in\pi_0(G)$ a principal bundle $P_{[g]}\in \mathrm{Bun}_G(S^1;*)$ with monodromy $[g]$.
A vector bundle $V$ over $\mathrm{Bun}_G(S^1;*)$ is then equivalent to a collection of representations $V_{[g]}$ of the groups $\mathrm{Aut}(P_{[g]})$.
The level $k\in H^4(BG,\Z)$ induces a gerbe over $\mathrm{Bun}_G(S^1;*)$,
which can equivalently be thought of as a collection of central extensions of the above isotropy groups.
We say that a $k$-twisted vector bundle over $\mathrm{Bun}_G(S^1;*)$ is \emph{locally normal} if its restriction to each $P_{[g]}\in\mathrm{Bun}_G(S^1;*)$ 
yields a locally normal representation of $\mathrm{Aut}(P_{[g]})\cong \Omega G$, in the sense we introduced.
We propose:
\begin{equation}\label{eq: T_G,k general}
\CS_{G,k}(\pt) = \big\{\parbox{9.8cm}{$k$-twisted locally normal vector bundles over $\mathrm{Bun}_G(S^1;*)$}\hspace{-.3mm}\big\},
\end{equation}
and conjecture that it is a bicommutant category.
This is equivalent (non-canonically) to the category of $\pi_0(G)$-tuples of locally normal representations of $\Omega G$ at level $k$.
%When $G$ is a disconnected group, we strongly suspect that the above tensor category is not of the form $T_\cA$ for any conformal net $\cA$.

The tensor structure on \eqref{eq: T_G,k general} is obtained by thinking about $G$-bundles over 
\tikz[scale=.5]{
\useasboundingbox (-1.2,-.05) rectangle (1.2,.85);
\draw circle (1) (-1,0) -- (1,0); \node[circle, fill=white, inner sep=0, scale =.8] at (1,0) {$*$};
},
trivialised at the base point $*$.
Given two locally normal vector bundles $V$ and \phantom{******} $W$
over $\mathrm{Bun}_G(S^1;*)$ and given a $G$-bundle $P$ over $S^1$, 
the value of $V\boxtimes W$ at $P$ is computed as follows.
Consider the finite set (indexed by $\pi_0(G)$) of isomorphism classes of extensions of $P$ over the above theta-graph,
and let $Q_i$ be representatives of the isomorphism classes.
Let $Q_i^+$ and $Q_i^-$ be the restrictions of $Q_i$ to the upper and lower halves of the theta-graph.
We identify those two halves
\tikz[scale=.4]{\useasboundingbox (-1.2,0) rectangle (1.2,.9);\draw (-1,0) arc (180:0:1) (-1,0) -- (1,0);\node[circle, fill=white, inner sep=0, scale =.75] at (1,0) {$*$};} and
\tikz[scale=.4]{\useasboundingbox (-1.2,-.2) rectangle (1.2,.7);\draw (-1,0) arc (-180:0:1) (-1,0) -- (1,0);\node[circle, fill=white, inner sep=0, scale =.75] at (1,0) {$*$};}
with $S^1$, so as to be able to view $Q_i^+$ and $Q_i^-$ as elements of $\mathrm{Bun}_G(S^1;*)$.

Let $G_0\subset G$ be the connected component of the identity.
Using a trivialisation of $Q_i$ over the middle edge of the theta-graph,
we get left and right actions of the algebra
$\cA_{G_0,k}(\tikz[scale=.5]{\useasboundingbox (-.5,-.2) rectangle (.5,.2);\draw (-.5,0) -- (.5,0);\draw[->] (.05,0) -- +(.01,0);})$ on the spaces $W(Q_i^+)$ and $V(Q_i^-)$ 
(the right action uses the identification
$\cA_{G_0,k}(\tikz[scale=.5]{\useasboundingbox (-.5,-.2) rectangle (.5,.2);\draw (-.5,0) -- (.5,0);\draw[->] (.05,0) -- +(.01,0);})^\op\cong
\cA_{G_0,k}(\tikz[scale=.5]{\useasboundingbox (-.5,-.2) rectangle (.5,.2);\draw (-.5,0) -- (.5,0);\draw[->] (-.05,0) -- +(-.01,0);})$).
The monoidal structure on the category of locally normal vector bundles over $\mathrm{Bun}_G(S^1;*)$
is then given by
\[
\big(V\boxtimes W\big)(P) := \bigoplus_i V(Q_i^-)\boxtimes_{\cA_{G_0,k}(\tikz[scale=.5]{\useasboundingbox (-.5,-.15) rectangle (.5,.25);\draw (-.5,0) -- (.5,0);\draw[->] (.05,0) -- +(.01,0);})}W(Q_i^+).
\]

The proposal \eqref{eq: T_G,k general} for $\CS(pt)$ unifies and generalizes the previous ones \eqref{eq: CS(pt)=Vect^k[G]} and \eqref{eq: CS(pt)=Rep(OmG)}.
It fits into the following commutative diagram:
\[
\tikzmath{
\node(A) at (0,0) {$\Bigg\{\,\parbox{2.4cm}{Unitary fusion\\\centerline{categories}}\,\Bigg\}$};
\node(B) at (5,0) {$\Bigg\{\,\parbox{2.2cm}{Bicommutant \\\centerline{categories}}\,\Bigg\}$};
\node(C) at (10.2,0) {$\Bigg\{\,\parbox{1.75cm}{Conformal \\\centerline{nets}}\,\Bigg\}$,};
\node(a) at (0,2.5) {$\left\{\parbox{2.6cm}{$ $ finite group $G$\\\centerline{+ level}\\\centerline{$k\in H^4(BG,\Z)$}}\right\}$};
\node(b) at (5,2.5) {$\left\{\parbox{3.6cm}{compact Lie group $G$\\\centerline{+ level}\\\centerline{$k\in H^4_+(BG,\Z)$}}\right\}$};
\node(c) at (10.2,2.5) {$\left\{\parbox{3.5cm}{$ $ compact connected\\\centerline{Lie group $G$\, {\footnotesize +} level}\\\centerline{$k\in H^4_+(BG,\Z)$}}\right\}$};
\draw[->] (A) --node[above, draw, circle, scale=.7, inner sep=2, yshift=4]{2} (B);
\draw[->] (C) --node[above, draw, circle, scale=.7, inner sep=2, yshift=4]{4} (B);
\draw[->, shorten <=5] (a) --node[pos=0, xshift=2.5, yshift=2.33]{$\scriptstyle \subset$} (b);
\draw[->, shorten <=5] (c) --node[pos=0, xshift=-2.5, yshift=2.33]{$\scriptstyle \supset$} (b);
\draw[->] (a) --node[left, draw, circle, scale=.7, inner sep=2, xshift=-4]{1} (A);
\draw[->] (b) --node[right, draw, circle, scale=.7, inner sep=2, xshift=4]{5} (B);
\draw[->] (c) --node[right, draw, circle, scale=.7, inner sep=2, xshift=4]{3} (C);
}
\]
where $H^4_+(BG,\Z)$ denotes the set of positive levels (Definition \ref{def: positive level}).
The arrow labelled~1 is the well known isomorphism $H^4(BG,\Z)\cong H^3(G,U(1))$.
The arrow labelled 2 is constructed in \cite{Bicommutant-categories-from-fusion-categories}.
The arrow labelled 3 is constructed in \cite[\S8]{WZW-classification}.
The arrow labelled 4 is constructed in \cite{Bicommutant-categories-from-conformal-nets} under the assumption that the conformal net has finite $\mu$-index.
The arrow labelled 5 is our proposal \eqref{eq: T_G,k general}, and
the fact that it produces a bicommutant category it is presently just a conjecture.

\section*{Acknowledgments}

I wish to thank Arthur Bartels and Chris Douglas for the long-time collaboration without which this work would not have been possible. 
I also thank Bruce Bartlett, David Ben-Zvi, Marcel Bischoff, Sebastiano Carpi, Dan Freed, Yi-Zhi Huang, Jacob Lurie, Hessel Posthuma, Chris Schommer-Pries, Urs Schreiber, Christoph Schweigert, Graeme Segal, Noah Snyder, Peter Teichner, Constantin Teleman, and Konrad Waldorf for many fruitful discussions, help, and advice.
Many thanks to MSRI for its hospitality during the spring of 2014.
At last, I gratefully acknowledge the Leverhulme trust, the EPSRC grant ``Quantum Mathematics and Computation'',
and the European Research Council under the European Union's Horizon 2020 research and innovation programme (grant agreement No 674978)
for financing my visiting position in Oxford.

\footnotesize
\bibliographystyle{abbrv}        \bibliography{CS(pt)}        
\def\cprime{$'$} \def\cprime{$'$} \def\cprime{$'$} \def\cprime{$'$}
\begin{thebibliography}{10}

\bibitem{MR1328736}
H.~H. Andersen and J.~Paradowski.
\newblock Fusion categories arising from semisimple {L}ie algebras.
\newblock {\em Comm. Math. Phys.}, 169(3):563--588, 1995.

\bibitem{MR1001453}
M.~Atiyah.
\newblock Topological quantum field theories.
\newblock {\em Inst. Hautes \'Etudes Sci. Publ. Math.}, (68):175--186 (1989),
  1988.

\bibitem{MR1100212}
S.~Axelrod, S.~Della~Pietra, and E.~Witten.
\newblock Geometric quantization of {C}hern-{S}imons gauge theory.
\newblock {\em J. Differential Geom.}, 33(3):787--902, 1991.

\bibitem{MR1355899}
J.~C. Baez and J.~Dolan.
\newblock Higher-dimensional algebra and topological quantum field theory.
\newblock {\em J. Math. Phys.}, 36(11):6073--6105, 1995.

\bibitem{MR1797619}
B.~Bakalov and A.~Kirillov, Jr.
\newblock {\em Lectures on tensor categories and modular functors}, volume~21
  of {\em University Lecture Series}.
\newblock American Mathematical Society, Providence, RI, 2001.

\bibitem{BDH-cn0}
A.~Bartels, C.~L. Douglas, and A.~Henriques.
\newblock Conformal nets and local field theory.
\newblock arXiv:0912.5307, 2010.

\bibitem{BDH-cn1}
A.~Bartels, C.~L. Douglas, and A.~Henriques.
\newblock Conformal nets {I}: {C}oordinate-free nets.
\newblock {\em Int. Math. Res. Not. IMRN}, (13):4975--5052, 2015.

\bibitem{BDH-cn4}
A.~Bartels, C.~L. Douglas, and A.~Henriques.
\newblock Conformal nets {IV}: the 3-category.
\newblock arXiv:1605.00662, 2016.

\bibitem{BDH-cn3}
A.~Bartels, C.~L. Douglas, and A.~Henriques.
\newblock Fusion of defects (formerly conformal nets {III}: fusion of defects).
\newblock {\em Memoirs of the AMS}, 2016.

\bibitem{BDH-cn2}
A.~Bartels, C.~L. Douglas, and A.~Henriques.
\newblock Conformal {N}ets {II}: {C}onformal {B}locks.
\newblock {\em Comm. Math. Phys.}, 354(1):393--458, 2017.

\bibitem{MR1360497}
A.~Beauville.
\newblock Conformal blocks, fusion rules and the {V}erlinde formula.
\newblock In {\em Proceedings of the {H}irzebruch 65 {C}onference on
  {A}lgebraic {G}eometry ({R}amat {G}an, 1993)}, volume~9 of {\em Israel Math.
  Conf. Proc.}, pages 75--96. Bar-Ilan Univ., Ramat Gan, 1996.

\bibitem{MR1289330}
A.~Beauville and Y.~Laszlo.
\newblock Conformal blocks and generalized theta functions.
\newblock {\em Comm. Math. Phys.}, 164(2):385--419, 1994.

\bibitem{arXiv:1410.8848}
M.~Bischoff, Y.~Kawahigashi, and R.~Longo.
\newblock Characterization of 2{D} rational local conformal nets and its
  boundary conditions: the maximal case.
\newblock arXiv:1410.8848, 2014.

\bibitem{MR1652746}
J.~B{\"o}ckenhauer and D.~E. Evans.
\newblock Modular invariants, graphs and {$\alpha$}-induction for nets of
  subfactors. {I}.
\newblock {\em Comm. Math. Phys.}, 197(2):361--386, 1998.

\bibitem{MR2174418}
A.~L. Carey, S.~Johnson, M.~K. Murray, D.~Stevenson, and B.-L. Wang.
\newblock Bundle gerbes for {C}hern-{S}imons and {W}ess-{Z}umino-{W}itten
  theories.
\newblock {\em Comm. Math. Phys.}, 259(3):577--613, 2005.

\bibitem{MR3039775}
A.~Davydov, M.~M{\"u}ger, D.~Nikshych, and V.~Ostrik.
\newblock The {W}itt group of non-degenerate braided fusion categories.
\newblock {\em J. Reine Angew. Math.}, 677:135--177, 2013.

\bibitem{MR1701599}
P.~Deligne and D.~S. Freed.
\newblock Classical field theory.
\newblock In {\em Quantum fields and strings: a course for mathematicians,
  {V}ol. 1, 2 ({P}rinceton, {NJ}, 1996/1997)}, pages 137--225. Amer. Math.
  Soc., Providence, RI, 1999.

\bibitem{MR1048699}
R.~Dijkgraaf and E.~Witten.
\newblock Topological gauge theories and group cohomology.
\newblock {\em Comm. Math. Phys.}, 129(2):393--429, 1990.

\bibitem{MR1245855}
C.~Dong.
\newblock Vertex algebras associated with even lattices.
\newblock {\em J. Algebra}, 161(1):245--265, 1993.

\bibitem{MR1233387}
C.~Dong and J.~Lepowsky.
\newblock {\em Generalized vertex algebras and relative vertex operators},
  volume 112 of {\em Progress in Mathematics}.
\newblock Birkh\"auser Boston, Inc., Boston, MA, 1993.

\bibitem{MR2261756}
C.~Dong and F.~Xu.
\newblock Conformal nets associated with lattices and their orbifolds.
\newblock {\em Adv. Math.}, 206(1):279--306, 2006.

\bibitem{arXiv:1312.7188}
C.~L. Douglas, C.~Schommer-Pries, and N.~Snyder.
\newblock Dualizable tensor categories.
\newblock arXiv:1312.7188, 2013.

\bibitem{MR1257326}
G.~Faltings.
\newblock A proof for the {V}erlinde formula.
\newblock {\em J. Algebraic Geom.}, 3(2):347--374, 1994.

\bibitem{MR1384612}
M.~Finkelberg.
\newblock An equivalence of fusion categories.
\newblock {\em Geom. Funct. Anal.}, 6(2):249--267, 1996.

\bibitem{MR3053762}
M.~Finkelberg.
\newblock Erratum to: {A}n equivalence of fusion categories.
\newblock {\em Geom. Funct. Anal.}, 23(2):810--811, 2013.

\bibitem{MR3330242}
D.~Fiorenza, H.~Sati, and U.~Schreiber.
\newblock A higher stacky perspective on {C}hern-{S}imons theory.
\newblock In {\em Mathematical aspects of quantum field theories}, Math. Phys.
  Stud., pages 153--211. Springer, Cham, 2015.

\bibitem{MR1273572}
D.~S. Freed.
\newblock Extended structures in topological quantum field theory.
\newblock In {\em Quantum topology}, volume~3 of {\em Ser. Knots Everything},
  pages 162--173. World Sci. Publ., River Edge, NJ, 1993.

\bibitem{MR1256993}
D.~S. Freed.
\newblock Higher algebraic structures and quantization.
\newblock {\em Comm. Math. Phys.}, 159(2):343--398, 1994.

\bibitem{MR1337109}
D.~S. Freed.
\newblock Classical {C}hern-{S}imons theory. {I}.
\newblock {\em Adv. Math.}, 113(2):237--303, 1995.

\bibitem{MR1898192}
D.~S. Freed.
\newblock Classical {C}hern-{S}imons theory. {II}.
\newblock {\em Houston J. Math.}, 28(2):293--310, 2002.
\newblock Special issue for S. S. Chern.

\bibitem{MR2476413}
D.~S. Freed.
\newblock Remarks on {C}hern-{S}imons theory.
\newblock {\em Bull. Amer. Math. Soc. (N.S.)}, 46(2):221--254, 2009.

\bibitem{MR2648901}
D.~S. Freed, M.~J. Hopkins, J.~Lurie, and C.~Teleman.
\newblock Topological quantum field theories from compact {L}ie groups.
\newblock In {\em A celebration of the mathematical legacy of {R}aoul {B}ott},
  volume~50 of {\em CRM Proc. Lecture Notes}, pages 367--403. Amer. Math. Soc.,
  Providence, RI, 2010.

\bibitem{MR2831111}
D.~S. Freed, M.~J. Hopkins, and C.~Teleman.
\newblock Loop groups and twisted {$K$}-theory {III}.
\newblock {\em Ann. of Math. (2)}, 174(2):947--1007, 2011.

\bibitem{MR1240583}
D.~S. Freed and F.~Quinn.
\newblock Chern-{S}imons theory with finite gauge group.
\newblock {\em Comm. Math. Phys.}, 156(3):435--472, 1993.

\bibitem{MR3334994}
D.~S. Freed and C.~Teleman.
\newblock Dirac families for loop groups as matrix factorizations.
\newblock {\em C. R. Math. Acad. Sci. Paris}, 353(5):415--419, 2015.

\bibitem{MR2082709}
E.~Frenkel and D.~Ben-Zvi.
\newblock {\em Vertex algebras and algebraic curves}, volume~88 of {\em
  Mathematical Surveys and Monographs}.
\newblock American Mathematical Society, Providence, RI, second edition, 2004.

\bibitem{MR1239440}
J.~Fr{\"o}hlich and T.~Kerler.
\newblock {\em Quantum groups, quantum categories and quantum field theory},
  volume 1542 of {\em Lecture Notes in Mathematics}.
\newblock Springer-Verlag, Berlin, 1993.

\bibitem{MR2076134}
J.~Fuchs, I.~Runkel, and C.~Schweigert.
\newblock T{FT} construction of {RCFT} correlators. {III}. {S}imple currents.
\newblock {\em Nuclear Phys. B}, 694(3):277--353, 2004.

\bibitem{MR1231644}
F.~Gabbiani and J.~Fr{\"o}hlich.
\newblock Operator algebras and conformal field theory.
\newblock {\em Comm. Math. Phys.}, 155(3):569--640, 1993.

\bibitem{MR1410566}
D.~Guido and R.~Longo.
\newblock The conformal spin and statistics theorem.
\newblock {\em Comm. Math. Phys.}, 181(1):11--35, 1996.

\bibitem{WZW-classification}
A.~Henriques.
\newblock The classification of chiral {WZW} models by ${H}^4_+({B}{G},\mathbb
  {Z})$.
\newblock {\em Contemporary Mathematics}, 2016.

\bibitem{Bicommutant-categories-from-conformal-nets}
A.~Henriques.
\newblock Bicommutant categories from conformal nets.
\newblock arXiv:1701.02052, 2017.

\bibitem{colimits}
A.~Henriques.
\newblock Loop groups and diffeomorphism groups of the circle as colimits.
\newblock arXiv:1706.08471, 2017.

\bibitem{Bicommutant-categories-from-fusion-categories}
A.~Henriques and D.~Penneys.
\newblock Bicommutant categories from fusion categories.
\newblock {\em Selecta Math. (N.S.)}, 23(3):1669--1708, 2017.

\bibitem{MR1065677}
N.~J. Hitchin.
\newblock Flat connections and geometric quantization.
\newblock {\em Comm. Math. Phys.}, 131(2):347--380, 1990.

\bibitem{MR2468370}
Y.-Z. Huang.
\newblock Rigidity and modularity of vertex tensor categories.
\newblock {\em Commun. Contemp. Math.}, 10(suppl. 1):871--911, 2008.

\bibitem{MR3339173}
Y.-Z. Huang, A.~Kirillov, Jr., and J.~Lepowsky.
\newblock Braided tensor categories and extensions of vertex operator algebras.
\newblock {\em Comm. Math. Phys.}, 337(3):1143--1159, 2015.

\bibitem{MR1327541}
Y.-Z. Huang and J.~Lepowsky.
\newblock Tensor products of modules for a vertex operator algebra and vertex
  tensor categories.
\newblock In {\em Lie theory and geometry}, volume 123 of {\em Progr. Math.},
  pages 349--383. Birkh\"auser Boston, Boston, MA, 1994.

\bibitem{MR3146015}
Y.-Z. Huang and J.~Lepowsky.
\newblock Tensor categories and the mathematics of rational and logarithmic
  conformal field theory.
\newblock {\em J. Phys. A}, 46(49):494009, 21, 2013.

\bibitem{MR696688}
V.~F.~R. Jones.
\newblock Index for subfactors.
\newblock {\em Invent. Math.}, 72(1):1--25, 1983.

\bibitem{MR1107651}
A.~Joyal and R.~Street.
\newblock Tortile {Y}ang-{B}axter operators in tensor categories.
\newblock {\em J. Pure Appl. Algebra}, 71(1):43--51, 1991.

\bibitem{MR1892455}
Y.~Kawahigashi.
\newblock Generalized {L}ongo-{R}ehren subfactors and {$\alpha$}-induction.
\newblock {\em Comm. Math. Phys.}, 226(2):269--287, 2002.

\bibitem{MR1838752}
Y.~Kawahigashi, R.~Longo, and M.~M{\"u}ger.
\newblock Multi-interval subfactors and modularity of representations in
  conformal field theory.
\newblock {\em Comm. Math. Phys.}, 219(3):631--669, 2001.

\bibitem{MR1104840}
D.~Kazhdan and G.~Lusztig.
\newblock Affine {L}ie algebras and quantum groups.
\newblock {\em Internat. Math. Res. Notices}, (2):21--29, 1991.

\bibitem{MR1186962}
D.~Kazhdan and G.~Lusztig.
\newblock Tensor structures arising from affine {L}ie algebras. {I}, {II}.
\newblock {\em J. Amer. Math. Soc.}, 6(4):905--947, 949--1011, 1993.

\bibitem{MR1239506}
D.~Kazhdan and G.~Lusztig.
\newblock Tensor structures arising from affine {L}ie algebras. {III}.
\newblock {\em J. Amer. Math. Soc.}, 7(2):335--381, 1994.

\bibitem{MR1239507}
D.~Kazhdan and G.~Lusztig.
\newblock Tensor structures arising from affine {L}ie algebras. {IV}.
\newblock {\em J. Amer. Math. Soc.}, 7(2):383--453, 1994.

\bibitem{MR829381}
H.~Kosaki.
\newblock Extension of {J}ones' theory on index to arbitrary factors.
\newblock {\em J. Funct. Anal.}, 66(1):123--140, 1986.

\bibitem{MR1289830}
S.~Kumar, M.~S. Narasimhan, and A.~Ramanathan.
\newblock Infinite {G}rassmannians and moduli spaces of {$G$}-bundles.
\newblock {\em Math. Ann.}, 300(1):41--75, 1994.

\bibitem{MR1669720}
Y.~Laszlo.
\newblock Hitchin's and {WZW} connections are the same.
\newblock {\em J. Differential Geom.}, 49(3):547--576, 1998.

\bibitem{MR1456243}
Y.~Laszlo and C.~Sorger.
\newblock The line bundles on the moduli of parabolic {$G$}-bundles over curves
  and their sections.
\newblock {\em Ann. Sci. \'Ecole Norm. Sup. (4)}, 30(4):499--525, 1997.

\bibitem{MR1273575}
R.~J. Lawrence.
\newblock Triangulations, categories and extended topological field theories.
\newblock In {\em Quantum topology}, volume~3 of {\em Ser. Knots Everything},
  pages 191--208. World Sci. Publ., River Edge, NJ, 1993.

\bibitem{Longo(Lectures-on-Nets-II)}
R.~Longo.
\newblock Lectures on conformal nets {II} \hfill.
\newblock http:/\!/www.mat.uniroma2.it/$\sim${}longo/Lecture\% 20Notes.html,
  2008.

\bibitem{MR1332979}
R.~Longo and K.-H. Rehren.
\newblock Nets of subfactors.
\newblock {\em Rev. Math. Phys.}, 7(4):567--597, 1995.
\newblock Workshop on Algebraic Quantum Field Theory and Jones Theory (Berlin,
  1994).

\bibitem{MR2100058}
R.~Longo and F.~Xu.
\newblock Topological sectors and a dichotomy in conformal field theory.
\newblock {\em Comm. Math. Phys.}, 251(2):321--364, 2004.

\bibitem{MR2555928}
J.~Lurie.
\newblock On the classification of topological field theories.
\newblock In {\em Current developments in mathematics, 2008}, pages 129--280.
  Int. Press, Somerville, MA, 2009.

\bibitem{MR1276910}
G.~Lusztig.
\newblock Monodromic systems on affine flag manifolds.
\newblock {\em Proc. Roy. Soc. London Ser. A}, 445(1923):231--246, 1994.

\bibitem{MR2105507}
G.~Lusztig.
\newblock Errata: ``{M}onodromic systems on affine flag manifolds'' [{P}roc.
  {R}oy. {S}oc. {L}ondon {S}er. {A} {\bf 445} (1994), no. 1923, 231--246;
  mr1276910].
\newblock {\em Proc. Roy. Soc. London Ser. A}, 450(1940):731--732, 1995.

\bibitem{MO:178113}
MathOverflow.
\newblock What's the state of affairs concerning the identification between
  quantum group reps at root of unity, and positive energy affine lie algebra
  reps?
\newblock http:/\!/mathoverflow.net/questions/178113/, 2014.

\bibitem{MR1690736}
E.~Meinrenken and C.~Woodward.
\newblock Hamiltonian loop group actions and {V}erlinde factorization.
\newblock {\em J. Differential Geom.}, 50(3):417--469, 1998.

\bibitem{MR1966525}
M.~M{\"u}ger.
\newblock From subfactors to categories and topology. {II}. {T}he quantum
  double of tensor categories and subfactors.
\newblock {\em J. Pure Appl. Algebra}, 180(1-2):159--219, 2003.

\bibitem{MR3055987}
T.~Nikolaus and K.~Waldorf.
\newblock Lifting problems and transgression for non-abelian gerbes.
\newblock {\em Adv. Math.}, 242:50--79, 2013.

\bibitem{Posthuma(PhD-thesis)}
H.~Posthuma.
\newblock Quantization of {H}amiltonian loop group actions.
\newblock {\em Ph.D. thesis, University of Amsterdam}, 2003.

\bibitem{MR2925299}
H.~Posthuma.
\newblock The {H}eisenberg group and conformal field theory.
\newblock {\em Q. J. Math.}, 63(2):423--465, 2012.

\bibitem{MR900587}
A.~Pressley and G.~Segal.
\newblock {\em Loop groups}.
\newblock Oxford Mathematical Monographs. The Clarendon Press, Oxford
  University Press, New York, 1986.
\newblock Oxford Science Publications.

\bibitem{MR2286123}
S.~F. Sawin.
\newblock Quantum groups at roots of unity and modularity.
\newblock {\em J. Knot Theory Ramifications}, 15(10):1245--1277, 2006.

\bibitem{MR2713992}
C.~J. Schommer-Pries.
\newblock The classification of two-dimensional extended topological field
  theories.
\newblock arXiv:1112.1000, 2014.

\bibitem{MR2079383}
G.~Segal.
\newblock The definition of conformal field theory.
\newblock In {\em Topology, geometry and quantum field theory}, volume 308 of
  {\em London Math. Soc. Lecture Note Ser.}, pages 421--577. Cambridge Univ.
  Press, Cambridge, 2004.

\bibitem{Staszkiewicz-thesis}
C.~P. Staszkiewicz.
\newblock Die lokale {S}truktur abelscher {S}tromalgebren auf dem {K}reis.
\newblock PhD thesis, Berlin, 1995.

\bibitem{MR1048605}
A.~Tsuchiya, K.~Ueno, and Y.~Yamada.
\newblock Conformal field theory on universal family of stable curves with
  gauge symmetries.
\newblock In {\em Integrable systems in quantum field theory and statistical
  mechanics}, volume~19 of {\em Adv. Stud. Pure Math.}, pages 459--566.
  Academic Press, Boston, MA, 1989.

\bibitem{MR2654259}
V.~G. Turaev.
\newblock {\em Quantum invariants of knots and 3-manifolds}, volume~18 of {\em
  de Gruyter Studies in Mathematics}.
\newblock Walter de Gruyter \& Co., Berlin, revised edition, 2010.

\bibitem{Waldorf-notes}
K.~Waldorf.
\newblock {C}hern-{S}imons theory and the categorified group ring.
\newblock available at http:/\!/ncatlab.org/nlab/files/WaldorfTalbot2010.pdf,
  2010.

\bibitem{MR2610397}
K.~Waldorf.
\newblock Multiplicative bundle gerbes with connection.
\newblock {\em Differential Geom. Appl.}, 28(3):313--340, 2010.

\bibitem{MR1645078}
A.~Wassermann.
\newblock Operator algebras and conformal field theory. {III}. {F}usion of
  positive energy representations of {${\rm LSU}(N)$} using bounded operators.
\newblock {\em Invent. Math.}, 133(3):467--538, 1998.

\bibitem{MR2443249}
S.~Willerton.
\newblock The twisted {D}rinfeld double of a finite group via gerbes and finite
  groupoids.
\newblock {\em Algebr. Geom. Topol.}, 8(3):1419--1457, 2008.

\bibitem{MR990772}
E.~Witten.
\newblock Quantum field theory and the {J}ones polynomial.
\newblock {\em Comm. Math. Phys.}, 121(3):351--399, 1989.

\bibitem{Wray-thesis}
K.~Wray.
\newblock Extended topological gauge theories in codimension zero and higher.
\newblock Master's Thesis, Universiteit van Amsterdam, 2010.

\bibitem{MR1776984}
F.~Xu.
\newblock Jones-{W}assermann subfactors for disconnected intervals.
\newblock {\em Commun. Contemp. Math.}, 2(3):307--347, 2000.

\end{thebibliography}
\end{document}